\RequirePackage{fix-cm}
\documentclass[smallextended]{svjour3}       
\smartqed  
\usepackage{graphicx}

\usepackage[utf8]{inputenc}
\usepackage{amsfonts}
\usepackage{amssymb}
\usepackage{natbib}
\usepackage{mathtools}
\usepackage{dsfont}
\usepackage{makeidx}
\usepackage{epstopdf}
\usepackage{float}
\usepackage{ragged2e}
\usepackage{url}
\usepackage{csquotes}

 \journalname{}

\newcommand{\PP}{{\mathbb P}}
\newcommand{\MP}{{\mathtt{MP}}}
\newcommand{\FS}{{\mathtt{FS}}}
\newcommand{\AL}{{\mathcal{A}}}
\newcommand{\E}{{\mathcal{E}}}

\begin{document}

\title{Quantifying the accuracy of ancestral state prediction in a phylogenetic tree under maximum parsimony
}


\author{Lina Herbst        \and
        Thomas Li \and
				Mike Steel
}


\institute{L. Herbst \at
              Institute of Mathematics and Computer Science, Ernst-Moritz-Arndt University,  Greifswald, Germany.
           \and
           T. Li \at
              School of Mathematics and Statistics, University of Canterbury, Christchurch, New Zealand.
					\and
					M. Steel \at
					Biomathematics Research Centre, University of Canterbury, Christchurch, New Zealand (corresponding author) \\
              \email{mike.steel@canterbury.ac.nz}
}

\date{Received: date / Accepted: date}

\maketitle

\begin{abstract}
In phylogenetic studies, biologists often wish to estimate the ancestral discrete character state at an interior vertex $v$ of an evolutionary tree $T$  from
the states that are observed at the leaves of the tree. A simple and fast estimation method --- maximum parsimony ---
takes the ancestral state at $v$ to be any state that minimises the number of state changes in $T$ required to explain
its evolution on $T$. In this paper, we investigate the reconstruction accuracy of this estimation method further, under 
a simple symmetric model of state change, and obtain a number of new results, both for 2-state characters, and $r$--state characters ($r>2$). Our results rely on establishing new identities and inequalities,
based on a coupling argument that involves a simpler `coin toss' approach to ancestral state reconstruction. 
\end{abstract}

\newpage

\section{Introduction}

Phylogenetic trees play a central role in evolutionary biology and in other related areas of classification (e.g. language evolution, stemmatology, ecology, epidemiology and medicine).  Typically, these trees represent a set of sampled `taxa' (e.g. species, genera, populations, individuals) as the leaves of the tree, with the vertices and edges of the tree providing a historical description of how these taxa evolved from a common ancestor \citep{fels4}.  Biologists often use discrete characteristics of the species at the leaves of a tree to try to infer (or predict) an ancestral state deep within the tree. For example, in epidemiology,  HIV sequences from sampled individuals have been used to estimate an ancestral form of the virus (e.g. for vaccine development) \citep{gas}; in  another study, ancestral state reconstruction played a key role in investigating the evolution of complex traits involved in animal vision, which varies across different species \citep{pla}.

Assuming that the characteristic in question has also evolved with the species, various methods have been devised to infer the ancestral state of that characteristic inside the tree and, in particular, at the last common ancestor of the species under study (i.e. the root of the tree). A method that can predict this root state allows any other ancestral vertex in the tree to also be studied, since one can re-root that tree on that vertex.
Thus, in this paper, we will assume that the root vertex is the one we wish to estimate an ancestral state for.
 
 The structure of this paper is as follows. First, we present some definitions concerning phylogenetic trees and a simple $r$--state Markovian model of character change on the tree, together with methods for predicting ancestral states, particularly maximum parsimony (MP).   In Section~\ref{secfun}, we concentrate on the 2-state model.  We describe an exact relationship between the reconstruction accuracies of MP on any binary tree $T$,  and  the accuracy on two trees derived from $T$ by deleting one and two leaves respectively.  We show how this allows inequalities to be established easily by induction. 
 
 
Next, in Section~\ref{secsim}, we describe a simpler ancestral prediction method that  is easier to analyse mathematically and yet is close enough to MP that it allows for  inequality results for MP to be established.   In particular, in Section~\ref{secrel}, we show that the reconstruction accuracy for this simple method is always a lower bound to MP under the 2-state model, thereby improving on existing known lower bounds.  In Section~\ref{rst}, we investigate the reconstruction accuracy for MP further in the more delicate setting when the number of states is greater than 2 and obtain some new inequality results.
In Section \ref{combo}, we present a novel combinatorial result that provides a sufficient condition for MP to infer the state at the root of a tree correctly, assuming only that the state changes in the tree are sufficiently well-spaced.  In the final section, we present a conjecture for future work.

\subsection{Definitions}

 In this paper, we consider rooted binary phylogenetic trees, which are trees in which every edge is directed away from a root vertex $\rho$ that has in-degree 0 and out-degree 1 or 2,
 and in which every non-root vertex has in-degree 1 and out-degree 0 or 2.  The vertices of out-degree 0 are the {\em leaves} of the tree.  In the case where $\rho$ has out-degree 2, we use $T$ to denote the tree, but if  $\rho$ has out-degree 1, we will  indicate this by writing $\dot{T}$ instead of $T$ and we will refer to the edge incident with this root as the {\em stem edge}.

Suppose that the root vertex $\rho$ has an associated state $F(\rho)$ that lies in some finite state space
$\AL$ of size $r \geq 2$, and that the root state evolves along the edges of the tree to the leaves according to a Markov process in which each edge $e$ has an associated probability $p_e$ of a change of state (called a {\em substitution}) between the endpoints of $e$. We refer to $p_e$ as the {\em substitution probability} for edge $e$. In this paper, we will assume that the underlying Markov process is the simple symmetric model on $r$ states, often referred to as the Neyman $r$--state model, denoted $N_r$. In this model, when a state change occurs on an edge $e=(u,v)$, each one of the  $r-1$ states that are different from the state at $u$ is assigned uniformly at random to the vertex $v$. In this way, each vertex $v$ of the tree
is assigned a random state, which we will denote as $F(v)$. 
We will denote the values of $F$ on the leaves of $T$ by the function  $f: X \rightarrow \AL$. 
This function $f=F|_X$ (the restriction of $F$ to the leaves of $T$)  is called a {\em character} in phylogenetics. Each such character has a well-defined probability under this stochastic model, and these probabilities sum to 1 over all the $r^n$ possible choices for $f$.

Given $f$, consider the set $\FS(f,T)$ of possible states that can be assigned to the root vertex of $T$ so as to minimise the total number of state changes required on the edges of $T$ to generate  $f$ at the leaves. The set $\FS(f,T)$ can be found in linear time (in $n$ and in $r$) by the first pass of the  `Fitch algorithm' \citep{fit71, har73}.  More precisely, to find $\FS(f, T)$, we assign a subset $\FS(v)$ of  $\AL$ to each vertex $v$ of $T$ in recursive fashion, starting from the leaves of $T$ and working towards the root vertex $\rho$ (we call $\FS(v)$ the {\em Fitch set} assigned to $v$).  First, each leaf $x$ is assigned the singleton set $\{f(x)\}$ as its Fitch set. Then for each vertex $v$ for which its two children $v_1$ and $v_2$ have  been assigned Fitch sets $\FS(v_1)$ and $\FS(v_2)$, respectively, the Fitch set $\FS(v)$ is determined as follows:
$$\FS(v)= \begin{cases}
\FS(v_1) \cap \FS(v_2), & \mbox{ if } \FS(v_1) \cap \FS(v_2) \neq \emptyset; \\
\FS(v_1) \cup \FS(v_2), &  \mbox{ if } \FS(v_1) \cap \FS(v_2) = \emptyset.
\end{cases}
$$
In this way, each vertex is eventually assigned a  non-empty subset of $\AL$ as its Fitch set, and 
 $\FS(f, T)$ is the Fitch set $\FS(\rho)$ that is assigned to the root vertex $\rho$.

When $\FS(f,T)$ consists of a single state,  then the method of {\em maximum parsimony} uses this state as the estimate of the unknown ancestral state
$\alpha$ at the root. When $\FS(f,T)$ has more than one state, we will select one of the states in this set uniformly at random as an estimate of the root state \citep{subset, moretaxa, ra1}.   
We will let $\MP(f,T)$ be the state selected uniformly at random from $\FS(f,T)$.

In this paper, we investigate the probability that this procedure correctly identifies the true root state $\alpha$ (note that by the symmetry in the model there is nothing special about
the choice of the root state $F(\rho$)).   We call this probability the {\em reconstruction accuracy} for maximum parsimony, denoted $RA_{\rm MP}(T)$. It is defined formally by: 

$$RA_{\rm MP}(T)  \coloneqq \PP(\MP(f, T) = \alpha | F(\rho) = \alpha).$$
Equivalently, 
  \begin{align} 
RA_{\rm MP}(T)= \sum_{\substack{\mathcal{R}: \mathcal{R} \subseteq \AL \\ \text{and } \alpha \in \mathcal{R}}} \frac{1}{| \mathcal{R} |} \cdot \PP(\FS(f,T)=\mathcal{R} | F(\rho)=\alpha). \label{RAdef}
\end{align}

\begin{figure}[ht]
\centering
\includegraphics[scale=0.7]{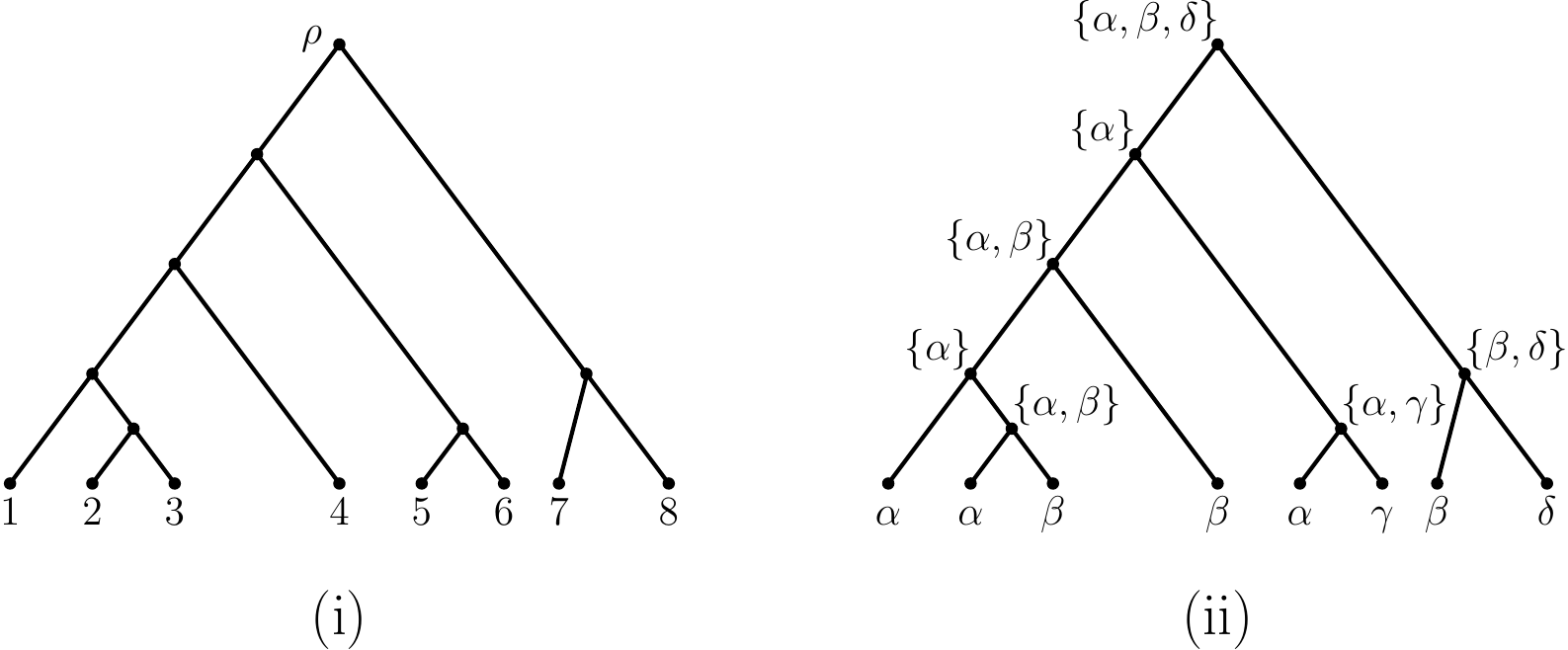}
\caption{(i) A rooted binary tree on leaf set $X=\{1,2,3,\ldots, 8\}$.  If we consider the character $f: X \rightarrow \AL=\{\alpha, \beta, \gamma,\delta\}$ defined by $f(1)=f(2)=f(5)=\alpha, f(3)=f(4)=f(7)=\beta, f(6)=\gamma, f(8)=\delta$, then the associated Fitch sets at the interior vertices are shown in (ii). 
Notice that the root Fitch set $\FS(f, T)$ consists of  three equally most-parsimonious root states, namely $\{\alpha, \beta, \delta\}$ and so $\MP(f, T)$ would be one of these states chosen with equal probability ($\frac{1}{3}$). An interesting feature of this example is that if the state of the leaf labelled 1 is changed from $\alpha$ to $\delta$,  then although $\delta$ was initially one of the most parsimonious states for the root, it  ceases to be so (instead,  $\beta$ becomes the unique most parsimonious root state).}
\label{fig1}
\end{figure}

 Because it is normally assumed that state changes occur according to an underlying continuous-time Markov  process, one has:
$$p_e \leq (r-1)/r.$$
We usually will  assume that this inequality is strict, since $p_e = (r-1)/r$ would correspond to an infinite rate of change (or an infinite temporal length) on the edge $e$ for a continuous-time Markov process. 
Given the  substitution probability $p_e$ for an edge $e$, we can formally associate a `length' for this edge as the quantity
 $\ell_e=-\frac{r-1}{r}\ln(1-\frac{r}{r-1}p_e)$. This `length' corresponds to the expected number of state changes under a continuous-time Markov-process realization of the substitution process (see e.g. \cite{fels4, ste16}).  Notice that we can write
 $p_e = \frac{r-1}{r}\left(1-\exp(-\frac{r}{r-1}\ell_e)\right)$. If we let $p(v)$ be the probability that vertex $v$ is in a different state from the root $\rho$ then
 $p(v) = \frac{r-1}{r}\left(1-\exp(-\frac{r}{r-1}L)\right)$, where $L$ is the sum of the $\ell$--lengths of the edges on the path from $\rho$ to $v$.

A special condition that is sometimes further imposed on these edge lengths  is that the edge lengths satisfy an {\em  ultrametric condition} (called a  `molecular clock' in biology),  which states that the sum of the lengths of the edges from the root to each leaf is the same.  Under that assumption, the probability $p(x)$ that leaf $x$ is in a different state from the root takes the same value for all values of $x$.  In this paper, our main results do not require this ultrametric assumption; however, we also point out how these results lead to particular conclusions in the ultrametric case. 

Note that $RA_{\rm MP}(T)$ depends on $T$, the assignment of state-change probabilities (the $p_e$ values) for the edges of $T$, and $r$ (the size of the state space $\AL$). 
The aim of this paper is to provide new relationships (equations and inequalities) for  reconstruction accuracy, extending earlier work by others \citep{her, moretaxa, ra1, subset}. 
Note that two other methods for estimating  the ancestral root state are {\em majority rule} (MR), which estimates the root state by the most frequently occurring state at the leaves
(ties are broken uniformly at random), and {\em maximum likelihood estimation} (MLE), which estimates the root state by the state(s) that maximise the probability of generating the given
character observed at the leaves. MR does not even require knowledge of the tree for estimating the root state, whereas MLE requires knowing not only the tree but also the edges lengths. 
Comparisons of these three methods were studied by \cite{gas10, gas14}.   If the edge lengths in MLE are not known,  and are therefore treated as `nuisance parameters' to be estimated (in addition to the root state) then  the resulting MLE estimate for the root state for a given character can be shown to be precisely the MP estimate under the $N_r$ model (\cite{tuf}, Theorem 6).

We end this section by collating some notation used throughout this paper.
\begin{itemize}
\item  $T$ (resp. $\dot{T}$) ---  a rooted binary tree, with a root of out-degree 2 (resp. out-degree 1),
\item $p_e$ (resp. $p_\rho$) --- the substitution probability on edge $e$, (resp. the stem edge of $\dot{T}$)  under the $N_r$ model,
\item $p(x)$ (resp. $p(w)$)  --- the probability that leaf $x$ (resp. vertex $w$) is in a different state from the root under the $N_r$ model,
\item $p_{\rm max}$ ---   the maximal value of $p(x)$ over all leaves,
\item $RA_{\rm MP}(T)$ --- the root-state reconstruction accuracy of maximum parsimony on $T$ (with its $p_e$ values) for a character generated under the $N_r$ model.
\end{itemize}

\section{A fundamental identity for reconstruction accuracy in the case where $r=2$. }
\label{secfun}

For Theorem \ref{ra} (below) we consider a rooted binary phylogenetic tree $T$ with a leaf set $X$ of size at least 3, together with two associated trees
$T'_\pi$ and $T''$  as indicated in Fig.~\ref{fig2}, which are determined by selecting a pair of leaves $y, z$ that are adjacent to a common vertex of $T$ (such a pair of leaves, called a `cherry', always exists in any binary tree with 3 or more leaves \citep{ste16}). 
The rooted binary phylogenetic tree $T'_\pi$ is obtained from $T$ by deleting the leaves $y$ and $z$; in addition, we lengthen the edge leading to $w$ slightly by putting an extra edge from $w$ to a new leaf $w'$ with substitution probability  $\pi$. In order to keep $T'_\pi$ binary, the vertex $w$ is suppressed.
An additional tree  $T''$ is obtained from $T'_\pi$ by deleting the edge leading to $w$ and edge $(w,w')$. Again, we suppress the resulting vertex of degree 2 in order to keep the tree binary. 

\begin{figure}[H]
\centering
\includegraphics[scale=0.85]{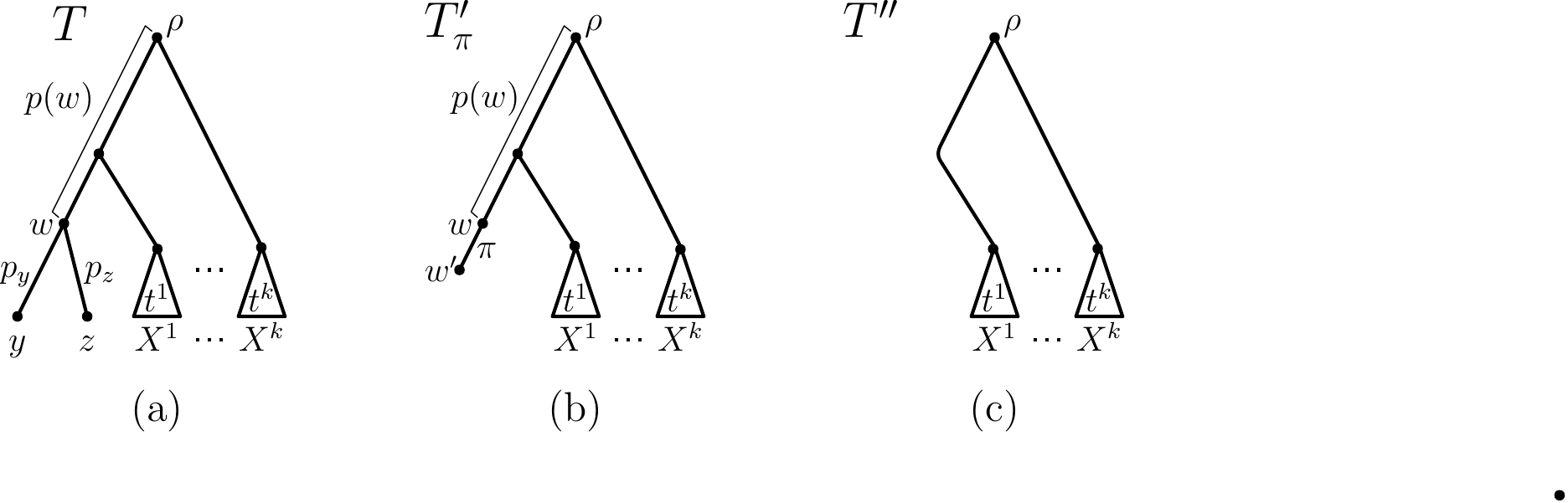}
\caption{(a) A rooted binary phylogenetic tree $T$ with leaf set $X$ where $p(w)$ is the probability that $w$ is in a different state from the root $\rho$,  
and $p_y$ and $p_z$ are the probabilities that leaves $y$ and $z$ are in a different state from $w$.  The pendant subtrees adjacent to the path from $w$ up to $\rho$ are denoted $t^1, \ldots, t^k$ with leaf sets $X^1, \ldots, X^k$, respectively. 
(b) A rooted binary phylogenetic tree $T'_\pi$ derived from $T$ by deleting leaves $y$ and $z$ and  attaching a new leaf $w'$ to $w$ (which is then suppressed). The value $\pi $ is the probability of a change of state from $w$ to the new leaf $w'$.  (c) The rooted binary tree $T''$ obtained from $T$ by deleting the leaves $y$ and $z$, their incident edges and the other edge incident with $w$, then suppressing the resulting vertex of degree 2. }
\label{fig2}
\end{figure}

We now state the main result of this section. Given $T$, $T'_\pi$ and $T''$  as described  we have the following fundamental equation for MP as ancestral state reconstruction method under the $N_2$ model.

\begin{theorem} 
\label{ra}
Let $T$ be a rooted binary phylogenetic tree with a leaf set $X$ of size at least 3.   For the reconstruction accuracy of maximum parsimony under the $N_2$ model we then have:
$$RA_{\rm MP}(T) = \theta \cdot RA_{\rm MP}(T'_\pi) + (1-\theta) \cdot RA_{\rm MP}(T''),$$
where $\theta$ is the probability that the leaves $y$ and $z$ are in the same state, and 
$\pi = p_yp_z/\theta \leq \min\{p_y, p_z\}$ (where $p_y$ and $p_z$ are the substitution probabilities for edges $(w,y)$ and $(w,z)$, respectively).

\end{theorem}

\begin{proof}
Let $T$ be a rooted binary phylogenetic tree with root $\rho$. By the symmetry in the model, we assume, without loss of generality, that  the root  is in state $\alpha$. 
Let  $\mathcal{F}$ denote the event that $\MP(f, T)=\alpha$ (recall that in the case of two equally-most-parsimonious states, one is selected uniformly at random). 
Let $\mathcal{E}_1$ be the event that leaf $y$ and leaf $z$ are in the same state (i.e. $f(y)=f(z)$), and let $\mathcal{E}_2$ be the complementary event (i.e. $f(y) \neq f(z)$). Thus  $\theta =  \PP(\mathcal{E}_1)$ and  $1- \theta = \PP(\mathcal{E}_2)$.  By the law of total probability we have:
\begin{align*}
RA_{\rm MP}(T) &= \PP(\mathcal{F}) = \PP(\mathcal{F}|\mathcal{E}_1) \PP(\mathcal{E}_1) + \PP(\mathcal{F}|\mathcal{E}_2) \PP(\mathcal{E}_2) \\
&= \PP(\mathcal{F}|\mathcal{E}_1) \theta + \PP(\mathcal{F}|\mathcal{E}_2) (1-\theta).
\end{align*}
We use this to establish Theorem~\ref{ra} by establishing the following two claims:
\begin{itemize}
\item[] Claim (i):  $RA_{\rm MP}(T'_\pi)$ equals $\PP(\mathcal{F}|\mathcal{E}_1)$;
\item[] Claim (ii):  $RA_{\rm MP}(T'')$ equals $\PP(\mathcal{F}|\mathcal{E}_2)$.
\end{itemize}

To establish Claim (i), we show that by an appropriate choice of $\pi$,  the probability that the leaves $y$ and $z$ are in state $\alpha$, conditional on the event $\mathcal{E}_1$, is exactly equal to the probability that $w'$ is in state $\alpha$; that is:
\begin{equation}
\label{ppeq}
\PP(f(y)=f(z)=\alpha |\mathcal{E}_1) = \PP(F(w')=\alpha).
\end{equation}
A similar equality will then hold for $\beta$ (i.e. 
$\PP(f(y)=f(z)=\beta |\mathcal{E}_1) = \PP(F(w')=\beta),$ since both probabilities sum up to 1).
These two identities then ensure that $RA_{\rm MP}(T'_\pi)$ equals $\PP(\mathcal{F}|\mathcal{E}_1)$, which is Claim (i). 
Thus for Claim (i), it suffices to establish Eqn.~(\ref{ppeq}) for a suitable choice of $\pi$. 


Recall that $1-p(w)$ is the probability that $w$ is in state $\alpha$, since the root is assumed to be in state $\alpha$. Then, the probability that $y$ and $z$ are in state $\alpha$ is 
\begin{align*}
\PP(f(y)=f(z)=\alpha) = (1-p(w))(1-p_y)(1-p_z)+p(w)p_yp_z,
\end{align*}
where $p_y$ and $p_z$ are the probabilities of change on edge $(w,y)$ and on edge $(w,z)$. Similarly, the probability that $y$ and $z$ are both in state $\beta$ is:
\begin{align*}
\PP(f(y)=f(z)=\beta) = p(w)(1-p_y)(1-p_z)+(1-p(w))p_yp_z.
\end{align*}
Adding these together, the probability of $\mathcal{E}_1$ is given by:
\begin{align*}
\PP(\mathcal{E}_1) &= (1-p(w))(1-p_y)(1-p_z)+p(w)p_yp_z + p(w)(1-p_y)(1-p_z)+(1-p(w))p_yp_z \\
&=(1-p_y)(1-p_z)+p_yp_z,
\end{align*}
which is independent of $p(w)$. \\
With substitution probability $\pi$ on edge $(w,w')$,  the probability that $w'$ is in state $\alpha$ is 
\begin{align}
\PP(F(w')=\alpha)=(1-p(w))(1-\pi) + p(w)\pi. \label{probP}
\end{align}

Now, 
\begin{align}
\PP(f(y)=f(z)=\alpha |\mathcal{E}_1)=\frac{(1-p(w))(1-p_y)(1-p_z)+p(w)p_yp_z}{\theta}, \label{probP2}
\end{align}
(recall that $\theta= \PP(\E_1)$).
We can write \eqref{probP2} as
\begin{align}
\PP(f(y)=f(z)=\alpha |\mathcal{E}_1)=(1-p(w))U+p(w)V, \label{probP3}
\end{align}
where $U=\frac{(1-p_y)(1-p_z)}{\theta}$ and $V=\frac{p_yp_z}{\theta}$ (note that $U+V=1$). Comparing \eqref{probP} with \eqref{probP3}, we see that if we take $\pi=V=\frac{p_yp_z}{\theta}$, then Eqn.~(\ref{ppeq}) (and hence Claim (i)) holds.

Notice also that with this choice, $\pi$ is less or equal to $p_y$ and to $p_z$.  For example, $\pi \leq p_y$ is equivalent to:
\begin{align*}
\pi&=\frac{p_yp_z}{(1-p_y)(1-p_z)+p_yp_z} \leq p_y \\
&\Leftrightarrow p_yp_z \leq p_y ((1-p_y)(1-p_z)+p_yp_z) \\
&\Leftrightarrow p_z \leq 1-p_z,
\end{align*}
which holds, since $p_z \leq \frac{1}{2}$.

To show that $RA_{\rm MP}(T'')=\PP(\mathcal{F}|\mathcal{E}_2)$, first notice that the probability of event $\mathcal{E}_2$ does not depend on the state at $w$ (i.e. $\PP(\mathcal{E}_2 | F(w)=\alpha)=\PP(\mathcal{E}_2 | F(w)=\beta)$), because:
\begin{align*}
\PP(\mathcal{E}_2) &= (1-p(w)) \big( (1-p_y)p_z + p_y(1-p_z) \big) + p(w) \big( (1-p_y)p_z + p_y(1-p_z) \big) \\
&= (1-p_y)p_z + p_y(1-p_z).
\end{align*}
Moreover, notice that when the leaves $y$ and $z$ take the states $\alpha, \beta$ (or $\beta, \alpha$) then the Fitch set for $w$ is $\{ \alpha, \beta \}$, so the state that is chosen as the ancestral state for $\rho$ is completely determined by the subtree $T''$. \\
Together with the argument above, this gives $RA_{\rm MP}(T'')=\PP(\mathcal{F}|\mathcal{E}_2)$, as required.
\hfill$\Box$
\end{proof}

Theorem \ref{ra} leads to the following corollary, which extends earlier results by \cite{subset} and by \cite{ra1} in which the  ultrametric constraint on the edge lengths was imposed (here this assumption is lifted).  

\begin{corollary} \label{1-p}
Let $T$ be a rooted binary phylogenetic tree with leaf set $X$.   Under the $N_2$ model: 
$$RA_{\rm MP}(T) \geq 1-p_{\rm max},$$
where $p_{\rm max}=\max\{p(x): x\in X\}$, and $p(x)$ is the probability that leaf $x$ has a different state from the root.
\end{corollary}

\begin{proof}
We use induction on the number of leaves $n$. For $n=1$,  $p_{\rm max}=p_x$ and thus the reconstruction accuracy is given by $RA_{\rm MP}(T)=1-p_{\rm max}$. For  $n=2$, and a tree with leaves $x, y$:
\begin{align*}
RA_{\rm MP}(T) &= (1-p_x)(1-p_y) + \frac{1}{2} \big( p_x (1-p_y) + (1-p_x) p_y \big) \\
&= 1-p_x-p_y+p_xp_y+ \frac{1}{2} \big( p_x -p_xp_y +p_y -p_xp_y \big) \\
&=1- \frac{1}{2} p_x - \frac{1}{2} p_y  \geq 1-p_{\rm max}.
\end{align*}
This completes the base case of the induction.

Now, assume that the claim holds for all rooted binary phylogenetic trees with less than $n$ leaves, where $n\geq 3$, and consider a tree with $n$ leaves represented as shown in Fig.~\ref{fig2}.
Let $p':= {\rm max} \{p(x): x \in (X\ \setminus \{y,z\}) \cup \{w\} \}$
and let
$p'':= {\rm max} \{p(x): x \in X \setminus \{y,z\}\}$.
Thus, $p',p'' \leq p_{\rm max}$ (the inequality for $p''$ is clear;  for $p'$ we use  $\pi \leq p_y, p_z$ from the last part of Theorem~\ref{ra}). Now,  from Theorem \ref{ra}, we have:
$$RA_{\rm MP}(T) = \theta \cdot RA_{\rm MP}(T'_\pi) + (1-\theta) \cdot RA_{\rm MP}(T''),$$
where $RA_{\rm MP}(T'_\pi) \geq 1-p'$ and $RA_{\rm MP}(T'') \geq 1-p''$ by the induction hypothesis. Thus:
\begin{align*}
RA_{\rm MP}(T) &\geq \theta (1-p') + (1- \theta) (1-p'') \\
&\geq \theta (1-p_{\rm max}) + (1- \theta) (1-p_{\rm max}) \\
&\qquad \text{since } p \geq p' \text{ and } p \geq p'' \\
&= (\theta + 1 - \theta) (1-p_{\rm max}) = 1-p_{\rm max}, 
\end{align*}
which completes the proof.
\hfill$\Box$
\end{proof}

\section{A `coin-toss' reconstruction method ($\varphi$)}
\label{secsim}
We now consider a method for estimating the ancestral state that is similar to the Fitch algorithm for MP, but which uses coin tosses to simplify the process.  We do this because it allows us to obtain results concerning MP by a coupling argument that relates MP  to this simpler method that is easier to analyse mathematically.   The coin-toss method works as follows: given a rooted binary phylogenetic tree $T$ and a character $f$ at the leaves of $T$, the method proceeds from the leaves to the root, just like
the Fitch algorithm described earlier. However, rather than assigning sets of states to each vertex, the coin toss method assigns a single state to each vertex.

More precisely, the coin-toss method starts (similarly to the Fitch algorithm) by assigning each leaf  the state given by the character $f$. For a vertex $v$ for which both direct descendants have been assigned states, if both these states are the same, then this state is also assigned to $v$. On the other hand, if the direct descendants have different states, than a fair coin is tossed to decide which of the two states to assign to $v$. This procedure  is continued upwards along the tree until the root is assigned a state.  We let  $\varphi$ denote this coin-toss method for ancestral state reconstruction, and denote the state selected by this method as $\varphi(T, f)$.  Let 
 $RA_{\varphi}(T)$ denote its reconstruction accuracy (i.e. the probability that it predicts the true root state in the $r$-state model, which equals
 $\PP(\varphi(T, f)= \alpha | F(\rho)= \alpha)$ for any state $\alpha$).

\begin{theorem} \label{coin-toss}
Let $T$ be a rooted binary phylogenetic tree with leaf set $X$. For $x \in X$, let $d(x)$ denote the number of edges between the root $\rho$ of $T$  and leaf $x$.  For the $N_r$ model (for any $r\geq 2$) we have: 
\begin{itemize}
\item[(i)]
$RA_{\varphi}(T) =  1- \sum_{x \in X} \Big( \frac{1}{2} \Big)^{d(x)} p(x);$
\item[(ii)]
$RA_{\varphi}(T) \geq 1-p_{\rm max}$,  and, 
\item[(iii)]  in the ultrametric setting, $RA_{\varphi}(T) = 1-p_{\rm max}$,
\end{itemize}
where $p_{\rm max}=\max\{p(x): x\in X\}$, and $p(x)$ is the probability that leaf $x$ has a different state from the root.
\end{theorem} 


\begin{proof}
{\em Part (i):}  Let $T$ be a rooted binary phylogenetic tree with root $\rho$ and leaf set $X$. Start at the root of $T$ and apply the following `reverse' process:  toss a fair coin and, depending on the outcome,  select one of the two children of  $\rho$ with equal probability. We keep going away from the root in this way until a leaf is reached. The root state is then estimated as the state at that leaf. Note that the reverse procedure (which proceeds from the root to the leaves) is stochastically identical in its estimated root state as the original coin-toss procedure $\varphi$.  Therefore, we have:
\begin{align}
RA_{\varphi}(T) &= \sum_{x \in X} \Big( \frac{1}{2} \Big)^{d(x)} (1-p(x))= 1- \sum_{x \in X} \Big( \frac{1}{2} \Big)^{d(x)} p(x), \label{raphi}
\end{align}
as claimed.  This establishes Part (i).

For Part (ii), with $p_{\rm max}=\max \{p(x): x \in X\}$, we have:
\begin{align*}
RA_{\varphi}(T) &= 1- \sum_{x \in X} \Big( \frac{1}{2} \Big)^{d(x)} p(x) &&\text{by } \eqref{raphi} \\
&\geq 1- p \sum_{x \in X} \Big( \frac{1}{2} \Big)^{d(x)} &&\text{because }  p\geq p(x) \text{ for all } x, \\
&= 1-p_{\rm max}, &&\text{because } \sum_{x \in X} \Big( \frac{1}{2} \Big)^{d(x)} =1,
\end{align*}
which gives Part (ii). \\
For Part (iii), we again observe that the reverse procedure for $\varphi$ is stochastically identical in its estimated root state to the coin toss procedure $\varphi$. 
Thus the reconstruction accuracy of $\varphi$ is just the probability that the leaf that is sampled  has the same state as the root, and this is clearly just $1-p_{\rm max}$ in case of an ultrametric tree and gives us Part $(iii)$ of the theorem.  
\hfill$\Box$
\end{proof}

Note that the reverse description of $\varphi$ should not be confused with the following even simpler estimation method: Select a leaf $x$ uniformly at random and estimate the ancestral root state by the state at $x$.  This method is stochastically equivalent to $\varphi$ only when $T$ is a complete balanced binary tree with $n=2^k$ leaves. In general, however, 
different leaves will have different probabilities of being chosen by the `reverse' description of $\varphi$, depending on the shape of the tree.

\subsection{Trees with a stem edge}

Shortly,  we will need to consider the reconstruction accuracy of a rooted binary tree $\dot{T}$ that has a root $\rho$ of out-degree 1, and so we pause to describe how this is related to the reconstruction accuracy of the tree $T$ adjacent to $\rho$.
Consider the stem edge leading from this degree-1 root $\rho$ to its child $\rho'$ and let $T$ be the tree obtained by removing this edge.  We can extend the definition of $RA_{\rm MP}$ and $RA_\varphi$ to $\dot{T}$ by simply assigning the predicted root state for $\rho'$ (for $T$) to the root $\rho$ of $\dot{T}$.  The following lemma describes a linear identity between the reconstruction accuracy of $\dot{T}$ and $T$ for MP and the coin-toss method $\varphi$.
\begin{lemma}
\label{lemhelps}
Under the $N_r$ model, suppose that the substitution probability for the stem edge $(\rho, \rho')$  of $\dot{T}$ is $p_\rho$.  If $M$ denotes either the method ${\rm MP}$ or $\varphi$, we then have:
$$RA_{M}(\dot{T})  = \left(1-\frac{r}{r-1}p_\rho\right)RA_M(T) + \frac{p_\rho}{r-1}.$$ 
\end{lemma}

\begin{proof}
By considering the two possible cases (no substitution on the stem edge, and a substitution to one of the $r-1$ non-root states),  the law of total probability gives:
\begin{equation}
\label{raa1}
RA_{M}(\dot{T}) = (1-p_\rho)\cdot RA_{M} (T) + p_\rho \cdot \PP(M(f, T) = \alpha|F(\rho') =\beta).
\end{equation}
for any state $\beta \neq \alpha$ (the choice does not matter because of the symmetry in the model). 
Now:
\begin{equation}
\label{raa2}
\sum_{\gamma \in \AL}\PP(M(f, T) = \gamma|F(\rho') = \beta) =1.
\end{equation}
The term on the left of this last equation can also be written as:
$$\PP(M(f, T)= \beta| F(\rho') =\beta) + \sum_{\gamma \neq \beta}  \PP(M(f, T) = \gamma|F(\rho') =\beta).$$
Moreover, the  $r-1$ probabilities in the summation term on the right of this last equation are all equal (again by the symmetries in the model). In particular, each of these $r-1$ probabilities is $\PP(M(f, T) = \alpha|F(\rho') =\beta)$.  
Combining this observation with Eqn.~(\ref{raa2}) gives:
$$1= RA_M(T) + (r-1)\PP(M(f, T) = \alpha|F(\rho') =\beta),$$
which rearranges to become:
$$ \PP(M(f, T) = \alpha | F(\rho') =\beta) = \frac{1-RA_M(T)}{r-1}.$$
Finally, substituting this expression into  Eqn.~(\ref{raa1}) gives the expression in the lemma.
\hfill$\Box$
\end{proof}

\subsection{Recursive equations for $RA_\varphi$}

We now consider a rooted binary phylogenetic tree $T$ with a root $\rho$ of out-degree 2, along with its two maximal pendant subtrees $T_1$ and $T_2$ with roots $\rho_1$ and $\rho_2$, respectively.
Let $\dot{T_1}$ be the tree obtained from $T$ by deleting $T_2$ and its incident edge $e$ and associated $p_e$ value (thus $\dot{T_1}$ is $T_1$ with the additional stem edge joining $\rho_1$ to $\rho$). Define $\dot{T_2}$ similarly, as indicated in Fig.~\ref{fig3}, and let  $p_i$ be the substitution probability for the  edge $(\rho, \rho_i)$. 
   
\begin{figure}[H]
\centering
\includegraphics[scale=0.9]{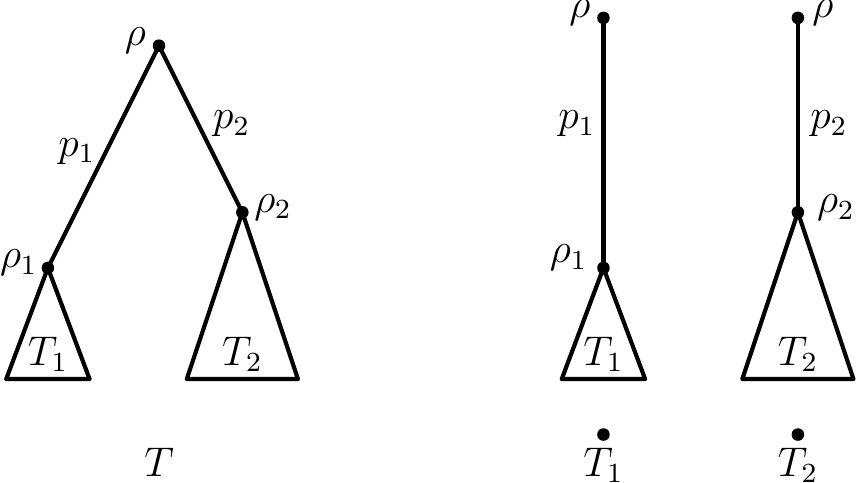}
\caption{{\em Left:} The tree $T$  with its two maximal subtrees $T_1$ and $T_2$. {\em Right:} The trees $\dot{T_1}$ and $\dot{T_2}$ obtained by attaching a stem edge to $T_1$ and $T_2$, with the same substitution probability $p_i$ as in $T$.}
\label{fig3}
\end{figure}

\begin{theorem} \label{coin-toss-average}
Let $T$ be a rooted binary phylogenetic tree with leaf set $X$.  Under the $N_r$ model,  the following identity for $RA_{\varphi}(T)$ holds:
\begin{align*}
RA_{\varphi}(T) = \frac{1}{2} \big( RA_{\varphi}(\dot{T_1}) + RA_{\varphi}(\dot{T_2}) \big).
\end{align*}
\end{theorem}     
\begin{proof}
Let $X_1$ and $X_2$ be the leaf sets of the trees $\dot{T}_1$ and $\dot{T}_2$.
For the reconstruction accuracy of the coin-toss method under the $N_r$ model, we then have:
\begin{align*}
RA_{\varphi}(T) &= 1- \sum_{x \in X} \Big( \frac{1}{2} \Big)^{d(x)} p(x) &&\qquad \text{by } \eqref{raphi} \\
&= 1- \sum_{x \in X_1} \Big( \frac{1}{2} \Big)^{d(x)} p(x) - \sum_{x \in X_2} \Big( \frac{1}{2} \Big)^{d(x)} p(x) \\
&= \frac{1}{2} - \sum_{x \in X_1} \Big( \frac{1}{2} \Big)^{d(x)} p(x) + \frac{1}{2} - \sum_{x \in X_2} \Big( \frac{1}{2} \Big)^{d(x)} p(x) \\
&= \frac{1}{2} \big(1 - \sum_{x \in X_1} \Big( \frac{1}{2} \Big)^{d(x)-1} p(x) \big) + \frac{1}{2} \big( 1 - \sum_{x \in X_2} \Big( \frac{1}{2} \Big)^{d(x)-1} p(x) \big) \\
&= \frac{1}{2} RA_{\varphi}(\dot{T_1}) + \frac{1}{2} RA_{\varphi}(\dot{T_2}) &&\qquad \text{by } \eqref{raphi} \\
&= \frac{1}{2} \big( RA_{\varphi}(\dot{T_1}) + RA_{\varphi}(\dot{T_2}) \big),
\end{align*}
which completes the proof.
\hfill$\Box$
\end{proof}

Similar to the fundamental equation for MP for ancestral state reconstruction, we have a fundamental equation for the coin-toss method given in Theorem \ref{fundamentaleqcoin} (this will be prove to be particularly useful in the next section). 
 For this equation, we consider $T'_{0.5}$ as depicted in Fig. \ref{fig2}, which is obtained from $T$ as in Fig. \ref{fig2} in the following way: Again, we delete the leaves $y$ and $z$. We then make the edge leading to $w$ infinitely long by putting an extra edge from $w$ to a new leaf $w^{\star}$ with the substitution probability $\pi=\frac{1}{2}$ on this edge. Setting $\pi=\frac{1}{2}$ simply means that both states are equally likely. Again, in order to keep the tree binary, vertex $w$ is suppressed.

Under the $N_2$ model, we have the following fundamental equation for the coin-toss method given $T, T'_\pi$ and $T'_{0.5}$ as described in Fig.~\ref{fig2} (note that  $T'_{0.5}$ is just $T'_\pi$ with $\pi=0.5$).

\begin{theorem} \label{fundamentaleqcoin}
Let $T$ be a rooted binary phylogenetic tree with leaf set $X$. Then, for the reconstruction accuracy of the coin-toss method under the $N_2$ model, we have:
$$RA_{\varphi}(T) = \theta \cdot RA_{\varphi}(T'_\pi) + (1-\theta) \cdot RA_{\varphi}(T'_{0.5}),$$
where $\theta$ and $\pi$ are as defined as in Theorem~\ref{ra}.
\end{theorem}
\begin{proof}
Let $T$ be a rooted binary phylogenetic tree with root $\rho$, and assume without loss of generality that the root is in state $\alpha$. We define $\mathcal{F}_{\varphi}$ to be the event that $\alpha$ is the state chosen for $\rho$ by the coin-toss method, and, as before,  let ${\mathcal E}_1$ be the probability that  leaves $y$ and $z$ have the same state. By the law of total probability:
\begin{align*}
RA_{\varphi}(T) &= \PP(\mathcal{F}_{\varphi}) = \PP(\mathcal{F}_{\varphi}|\mathcal{E}_1) \PP(\mathcal{E}_1) + \PP(\mathcal{F}_{\varphi}|\mathcal{E}_2) \PP(\mathcal{E}_2) \\
&= \PP(\mathcal{F}_{\varphi}|\mathcal{E}_1) \theta + \PP(\mathcal{F}_{\varphi}|\mathcal{E}_2) (1-\theta).
\end{align*}
In order to prove Theorem \ref{fundamentaleqcoin}, it remains to show $RA_{\varphi}(T'_\pi)=\PP(\mathcal{F}_{\varphi}|\mathcal{E}_1)$ and $RA_{\varphi}(T'_{0.5})=\PP(\mathcal{F}_{\varphi}|\mathcal{E}_2)$ respectively.
Now, $RA\varphi(T'_\pi) = {\mathbb P}({\mathcal F}_{\varphi}|\E_1)$ since, conditional on $\E_1$, the state chosen by $\varphi$ at $w$ in $T$ has the same probability distribution as the state chosen by $\varphi$ at $w'$ in $T'_\pi$, and the remainder of application of $\varphi$ to $T$ and $T'_\pi$ is identical.
\\
We have $RA_{\varphi}(T'_{0.5})=\PP(\mathcal{F}_{\varphi}|\mathcal{E}_2)$, because by having the substitution probability $\pi=\frac{1}{2}$ for the edge leading to $w^{\star}$ both states $\alpha$ and $\beta$ are equally probable. So the probability of choosing $\alpha$ for $w$ is $\frac{1}{2}$. Moreover, on $T$ the probability of choosing $\alpha$ for $w$ from the states at the leaves $y$ and $z$ conditional on event $\mathcal{E}_2$ (i.e. $y$ and $z$ are in different states), is $\frac{1}{2}$ as well. 
\hfill$\Box$
\end{proof}

\section{The relationship between the two ancestral reconstruction methods}
\label{secrel}

\begin{theorem} \label{mpandcoin}
Let $T$ be a rooted binary phylogenetic tree with leaf set $X$. Under the $N_2$ model, the reconstruction accuracy of MP is at least equal to the reconstruction accuracy of the coin-toss method; that is:
\begin{align*} 
RA_{\rm MP}(T) \geq RA_{\varphi}(T).
\end{align*}
\end{theorem}
\begin{proof}
The proof is by induction on the number of leaves. For $n=2$ and a tree with leaves $x,y$, we have:
\begin{align*}
RA_{\rm MP}(T) &= (1-p_x)(1-p_y) + \frac{1}{2} (p_x(1-p_y)+(1-p_x)p_y) \\
&= \frac{1}{2} \big( (1-p_x) + (1-p_y) \big).
\end{align*}
By Theorem \ref{coin-toss-average}, the reconstruction accuracy of the coin-toss method is exactly the average of the reconstruction accuracy of both subtrees. Therefore, $RA_{\varphi}(T)=\frac{1}{2} \big( (1-p_x) + (1-p_y) \big)$, which is equal to $RA_{\rm MP}(T)$, and establishes the base case of the induction. \\
Now assume that the induction hypothesis holds for all rooted binary phylogenetic trees with fewer than $n$ leaves, where $n \geq 3$.  By Theorem \ref{ra}, we have:
\begin{align*}
RA_{\rm MP}(T) = \theta \cdot RA_{\rm MP}(T'_\pi) + (1-\theta) \cdot RA_{\rm MP}(T''),
\end{align*}
with $T'_\pi$ and $T''$ as in Fig.~\ref{fig2}, and $\theta$ as described above. Additionally, by Theorem \ref{fundamentaleqcoin} we have that
\begin{align*}
RA_{\varphi}(T) = \theta \cdot RA_{\varphi}(T'_\pi) + (1-\theta) \cdot RA_{\varphi}(T'_{0.5}),
\end{align*}
with $T'_\pi$ and $T'_{0.5}$ as in Fig.~\ref{fig2}. By the induction hypothesis, $RA_{\rm MP}(T'_\pi) \geq RA_{\varphi}(T'_\pi)$ and $RA_{\rm MP}(T'') \geq RA_{\varphi}(T'')$ both hold, so in order to complete the proof, it remains to show that $RA_{\varphi}(T'') \geq RA_{\varphi}(T'_{0.5})$.  The intuition behind this inequality is that when the leaf at the end of a pendant edge is completely random (i.e. no more likely to match the root state than not match it) then pruning this edge cannot reduce the reconstruction  accuracy of $\varphi$. 
Note that $T''$ has one leaf fewer than $T'_{0.5}$. In the following we consider both trees as shown in Fig.~\ref{fig2}. As before, all vertices of degree 2 are suppressed to keep the tree binary. In order to calculate $RA_{\varphi}(T'')$ and $RA_{\varphi}(T'_{0.5})$, consider  the subtrees $t^{1}, \dots, t^{k}$ of $T$ and their corresponding leaf sets $X^{1}, \dots, X^{k}$, that are adjacent to the path from $w$ up to $\rho$. These leaf sets partition the leaf set of $T''$, and if we add in the additional set $\{w^{\star}\}$, then this collection of $k+1$ sets partitions the leaves of $T'_\pi$ and $T'_{0.5}$.
By Theorem~\ref{coin-toss}(i) we have:
\begin{align*}
RA_{\varphi}(T'') &= 1- \sum_{x \in X^{1} \cup \dots \cup X^{k}} \Big( \frac{1}{2} \Big)^{d''(x)} p(x),
\end{align*}
where $d''(x)$ is the number of edges between the root and a leaf $x$ in $T''$, and:
\begin{align*}
RA_{\varphi}(T'_{0.5}) &= 1- \sum_{x \in X^{1} \cup \dots \cup X^{k} \cup \{w^{\star}\}} \Big( \frac{1}{2} \Big)^{d'(x)} p(x),
\end{align*}
where $d'(x)$ is the number of edges between the root and leaf $x$ in $T'_{0.5}$. Moreover, note that for $i=2, \dots,k$:
\begin{align*}
\sum_{x \in X^{i}} \Big( \frac{1}{2} \Big)^{d''(x)} p(x) = \sum_{x \in X^{i}} \Big( \frac{1}{2} \Big)^{d'(x)} p(x).
\end{align*} 
Thus, $RA_{\varphi}(T'') - RA_{\varphi}(T'_{0.5})$ becomes:
\begin{align*}
RA_{\varphi}(T'') - RA_{\varphi}(T'_{0.5}) &= 1- \sum_{x \in X^{1}} \Big( \frac{1}{2} \Big)^{d''(x)} p(x) - \big(1-\sum_{x \in X^{1}} \Big( \frac{1}{2} \Big)^{d'(x)} p(x) - \Big(\frac{1}{2}\Big)^{k}p({w^{\star}}) \big) \\
&= \sum_{x \in X^{1}} \Big( \frac{1}{2} \Big)^{d'(x)} p(x) - \sum_{x \in X^{1}} \Big( \frac{1}{2} \Big)^{d''(x)} p(x) + \Big(\frac{1}{2}\Big)^{k}p({w^{\star}}).
\end{align*}
We have $\pi=\frac{1}{2}$, which gives us $p({w^{\star}})=p(w)+\pi-2p(w)\frac{1}{2}=p(w)+\frac{1}{2}-p(w)=\frac{1}{2}$. Again, note that the vertex $w$ is suppressed in $T'_{0.5}$ in order to keep the tree binary, and thus $k$ edges separate the root and the leaf $w^{\star}$. Similarly, the vertex leading to subtree $t^{1}$ is suppressed in $T''$ to keep the tree binary. This gives us that $k-1$ edges separate the root of $T''$ and the root of $t^{1}$, whereas $k$ edges separate the root of $T'_{0.5}$ and the root of the subtree $t^{1}$. Let $d_1(x)$ denote the number of edges between the root of subtree $t^{1}$ and leaf  $x$ in $T$,  then we have
$$d'(x)=k+d_1(x) \mbox{ and } d''(x) = k-1 + d_1(x).$$  If we now rearrange the above expression for  $RA_{\varphi}(T'') - RA_{\varphi}(T'_{0.5})$, noting that $p(w^*)=\frac{1}{2}$ we obtain:
\begin{align*}
RA_{\varphi}(T'') - RA_{\varphi}(T'_{0.5}) &= \sum_{x \in X^{1}} \Big( \frac{1}{2} \Big)^{k+d_1(x)} p(x) - \sum_{x \in X^{1}} \Big( \frac{1}{2} \Big)^{k-1+d_1(x)} p(x) + \Big(\frac{1}{2}\Big)^{k+1}, \\
&= \Big( \frac{1}{2} \Big)^{k} \sum_{x \in X^{1}} \Big( \frac{1}{2} \Big)^{d_1(x)} p(x) - \Big( \frac{1}{2} \Big)^{k-1} \sum_{x \in X^{1}} \Big( \frac{1}{2} \Big)^{d_1(x)} p(x) + \Big(\frac{1}{2}\Big)^{k+1} \\
&= - \Big( \frac{1}{2} \Big)^{k} \sum_{x \in X^{1}} \Big( \frac{1}{2} \Big)^{d_1(x)} p(x) + \Big(\frac{1}{2}\Big)^{k+1} \\
&=\Big(\frac{1}{2}\Big)^{k+1} \Big( 1-2 \sum_{x \in X^{1}} \Big( \frac{1}{2} \Big)^{d_1(x)} p(x) \Big) \\
&\geq 0,
\end{align*}
since for all $x \in X$, we have $d_1(x) \geq 0$ and $0 \leq p(x) \leq \frac{1}{2}$ in the $N_2$ model. Therefore, we have $RA_{\varphi}(T'') \geq RA_{\varphi}(T'_{0.5})$, which, together with the induction hypothesis, gives $RA_{\rm MP}(T'') \geq RA_{\varphi}(T'') \geq RA_{\varphi}(T'_{0.5})$ and thus completes the proof.
\hfill$\Box$
\end{proof}
Note that combining the statement of Theorem \ref{mpandcoin} with Theorem~\ref{coin-toss} gives us an alternative proof of Corollary~\ref{1-p}, since  $RA_{\rm MP}(T) \geq RA_{\varphi}(T) \geq 1-p_{\rm max}$ (i.e. under the $N_2$ model, the Fitch algorithm using all terminal taxa is at least as accurate for ancestral state reconstruction as selecting the state of a taxon $x$ that maximises $p(x)$).

\section{Further results for the $r$--state setting}
\label{rst}

In this section, we will indicate the set of states in $\AL$ by writing $\AL=\{\alpha_1, \alpha_2, \ldots, \alpha_r\}$, and, unless stated otherwise, we assume the root is in state $\alpha_1$.
For a set $\mathcal{R} \subseteq \AL$, $\alpha_1 \in \mathcal{R}, |\mathcal{R} |=k$,  let 
$$P_k(T) \coloneqq \PP(\FS(f,T)=\mathcal{R} | F(\rho)=\alpha_1).$$ 
Similarly,  for a set $\mathcal{R} \subseteq \AL$, $\alpha_1 \notin \mathcal{R}, |\mathcal{R} |=k$, let 
$$Q_k(T) \coloneqq \PP(\FS(f,T)=\mathcal{R} | F(\rho)=\alpha_1).$$ 
By the symmetry in the model, the values $P_k(T)$ and $Q_k(T)$ are independent of the choice of $\mathcal{R}$, subject to the constraints imposed on $\mathcal{R}$ in their definition.

\begin{lemma} \label{rstatesRA}
For any rooted binary phylogenetic tree $T$ under the $N_r$ model, the reconstruction accuracy of MP is given by:
\begin{align*}
RA_{\rm MP}(T)= \frac{1}{r} \Big(1 + \sum_{k=1}^{r-1} {r-1 \choose k} (P_k(T) - Q_k(T)) \Big).
\end{align*}

\end{lemma}
\begin{proof}
Let $T$ be a rooted binary phylogenetic tree and let $P_k(T)$ and $Q_k(T)$ be as defined above (so we assume the root to be in state $\alpha_1$). For the reconstruction accuracy of MP under the $N_r$ model, Eqn. \eqref{RAdef}  and the law of total probability gives:
\begin{align*}
RA_{\rm MP}(T) &= \sum_{\substack{\mathcal{R}: \mathcal{R} \subseteq \AL \\ \text{and } \alpha_1 \in \mathcal{R}}} \frac{1}{| \mathcal{R} |} \cdot \PP(\FS(f,T)=\mathcal{R} | F(\rho)=\alpha_1) \\
&= \sum_{k=1}^{r} \frac{1}{k} ~P_k(T)~ {r-1 \choose k-1} = \sum_{k=1}^{r-1} \frac{1}{k} ~P_k(T)~ {r-1 \choose k-1} + \frac{1}{r} ~p_r \\
&= \sum_{k=1}^{r-1} \frac{1}{k} ~P_k(T)~ {r-1 \choose k-1} + \frac{1}{r} \big(1-\sum_{k=1}^{r-1} Q_k(T)~ {r-1 \choose k} -\sum_{k=1}^{r-1} P_k(T)~ {r-1 \choose k-1} \big). \\
\end{align*}
Rearranging this last expression gives:
\begin{align*}
RA_{\rm MP}(T) & = \frac{1}{r} + \sum_{k=1}^{r-1} \big(\frac{1}{k} - \frac{1}{r} \big) ~P_k(T)~ {r-1 \choose k-1} - \frac{1}{r} \sum_{k=1}^{r-1} Q_k(T)~ {r-1 \choose k} \\
&= \frac{1}{r} + \frac{1}{r} \sum_{k=1}^{r-1} \frac{r-k}{k} ~P_k(T)~ {r-1 \choose k-1} - Q_k(T)~ {r-1 \choose k} \\
&= \frac{1}{r} + \frac{1}{r} \sum_{k=1}^{r-1} {r-1 \choose k} (P_k(T)- Q_k(T)) \\
&=\frac{1}{r} \Big(1 + \sum_{k=1}^{r-1} {r-1 \choose k} (P_k(T) - Q_k(T)) \Big).
\end{align*}
\hfill$\Box$
\end{proof}

For the following lemma we consider $\dot{T}$ obtained from $T$ by adding an additional stem edge $(\rho,\rho')$ and substitution probability $p_{\rho}$ on this edge.
Let 
\begin{align*}
P_{\alpha_1, \dots, \alpha_k}(\dot{T}) &\coloneqq \PP(\FS(f,\dot{T})=\{\alpha_1\dots, \alpha_k\} |F(\rho)= \alpha_1) \mbox{ and } \\
P_{\alpha_2, \dots, \alpha_{k+1}}(\dot{T}) &\coloneqq \PP(\FS(f,\dot{T})=\{\alpha_2\dots, \alpha_{k+1}\} |F(\rho)= \alpha_1)
\end{align*}

\begin{lemma} \label{rstatesprob}
Assume that $\rho$ is in state $\alpha_1$. Under the $N_r$ model and $1 \leq k \leq r-1$, we have:
\begin{align*}
P_{\alpha_1, \dots, \alpha_k}(\dot{T})  =   (1- \frac{r-k}{r-1}p_{\rho})~ P_k(T) + \frac{r-k}{r-1}~ p_\rho Q_k(T), \mbox{ and }\\
P_{\alpha_2, \dots, \alpha_{k+1}}(\dot{T})  =  (1- \frac{k}{r-1}p_{\rho})~ Q_k(T) + \frac{k}{r-1} p_\rho P_k(T),
\end{align*}
where $P_k(T)$ and $Q_k(T)$ are as defined above.
\end{lemma}

\begin{proof}
For $1 \leq k \leq r-1$, we can write $P_{\alpha_1, \dots, \alpha_k}(\dot{T})$ as follows:
\begin{equation}
\label{ry1}
(1-p_\rho) \PP(\FS(f, T)=\{\alpha_1,\dots, \alpha_k\} |F(\rho') = \alpha_1, F(\rho)= \alpha_1) + \frac{p_{\rho}}{r-1} S,
\end{equation}
where 
$$S= \sum_{i=2}^r \PP(\FS(f,T)=\{\alpha_1,\dots, \alpha_k\} |F(\rho')=\alpha_i, F(\rho)= \alpha_1).$$
We can now split $S$ into two sums depending on the range of $k$. Thus we have $S=S_1 + S_2$, where:
$$S_1= \sum_{i=2}^k \PP(\FS(f,T)=\{\alpha_1,\dots, \alpha_k\} |F(\rho')=\alpha_i, F(\rho)= \alpha_1), \mbox{ and }$$
$$S_2=\sum_{i=k+1}^r \PP(\FS(f,T)=\{\alpha_1,\dots, \alpha_k\} |F(\rho')=\alpha_i, F(\rho)= \alpha_1).$$
Notice also that, by the symmetry of the model,  each of the $k-1$ terms  in $S_1$  is equal to $$\PP(\FS(f,T)=\{\alpha_1,\dots, \alpha_k\} |F(\rho')=\alpha_1, F(\rho)= \alpha_1),$$
which is $P_k(T)$. Thus $S_1 = (k-1)P_k(T)$.

Similarly, each of the $r-k$ terms in $S_2$ is equal to 
$$\PP(\FS(f,T)=\{\alpha_2,\dots, \alpha_{k+1}\} |F(\rho')=\alpha_1, F(\rho)= \alpha_1),$$
which is just $Q_k(T)$, and thus $S_2 = (r-k)Q_k(T)$.  Thus, from the expression for $P_{\alpha_1, \dots, \alpha_k}(\dot{T})$ given by (\ref{ry1}), we have:
$$P_{\alpha_1, \dots, \alpha_k}(\dot{T}) = (1-p_\rho)P_k(T) + \frac{p_{\rho}}{r-1} ((k-1)P_k(T) + (r-k)Q_k(T)).$$
Rearranging the term on the right gives the expression for $P_{\alpha_1, \dots, \alpha_k}(\dot{T})$ in Lemma~\ref{rstatesprob}.


The second part of Lemma~\ref{rstatesprob} follows by an analogous argument.  For $1 \leq k \leq r-1$, we can write
$P_{\alpha_2, \dots, \alpha_{k+1}}(\dot{T})$ as follows:
\begin{equation}
\label{ry2}
(1-p_\rho) \PP(\FS(f,T)=\{\alpha_2,\dots, \alpha_{k+1}\} | F(\rho')=\alpha_1, F(\rho)= \alpha_1)  + \frac{p_{\rho}}{r-1} S',
\end{equation}
where 
$$S'= \sum_{i=2}^r \PP(\FS(f,T)=\{\alpha_2,\dots, \alpha_{k+1}\} |F(\rho')=\alpha_i, F(\rho)= \alpha_1).$$
Write $S'=S'_1 + S'_2$ where:
$$S'_1= \sum_{i=2}^k \PP(\FS(f,T)=\{\alpha_2,\dots, \alpha_{k+1}\} |F(\rho')=\alpha_i, F(\rho)= \alpha_1), \mbox{ and }$$
$$S'_2=\sum_{i=k+1}^r \PP(\FS(f,T)=\{\alpha_2,\dots, \alpha_{k+1}\}  |F(\rho')=\alpha_i, F(\rho)= \alpha_1).$$
Notice also that, by the symmetry of the model,  each of the $k$ terms  in $S'_1$  is equal to $$\PP(\FS(f,T)=\{\alpha_1,\dots, \alpha_k\} |F(\rho')=\alpha_1, F(\rho)= \alpha_1),$$
which is $P_k(T)$. Thus $S'_1 = kP_k(T)$.

Similarly, each of the $r-k-1$ terms in $S'_2$ is equal to 
$$\PP(\FS(f,T)=\{\alpha_2,\dots, \alpha_{k+1}\} |F(\rho')=\alpha_1, F(\rho)= \alpha_1),$$
which is just $Q_k(T)$, and thus $S'_2 = (r-k-1)Q_k(T)$.  Thus, from the expression for
 $P_{\alpha_2, \dots, \alpha_{k+1}}(\dot{T})$ given by (\ref{ry2}) we have:
$$P_{\alpha_2, \dots, \alpha_{k+1}}(\dot{T}) = (1-p_\rho)Q_k(T) + \frac{p_{\rho}}{r-1} (kP_k(T) + (r-k-1)Q_k(T)).$$
Rearranging the term on the right gives the expression for $P_{\alpha_2, \dots, \alpha_{k+1}}(\dot{T})$ in Lemma~\ref{rstatesprob}.
\hfill$\Box$
\end{proof}

By the proof of Lemma \ref{rstatesprob}, we have the following corollary.
\begin{corollary} \label{corrstates}
Let $\dot{T}$ be a rooted binary phylogenetic tree with  stem edge $(\rho,\rho')$.
Consider the $N_r$ model with state space $\AL=\{\alpha_1, \dots, \alpha_r\}$, assume the root $\rho$ is in state $\alpha_1$, and let $p_{\rho}$ be the substitution probability on the stem edge.
Then, for $1 \leq k \leq r-1$ we have:
\begin{align*}
&(i)~ P_{\alpha_1, \dots, \alpha_k}(\dot{T}) - P_{\alpha_2, \dots, \alpha_{k+1}}(\dot{T}) = (1-\frac{r}{r-1}~p_{\rho}) (P_k(T) - Q_k(T)) \\
&(ii)~ P_{\alpha_1, \dots, \alpha_k}(\dot{T}) = P_{\alpha_2, \dots, \alpha_{k+1}}(\dot{T}) + (1-\frac{r}{r-1}~p_{\rho}) (P_k(T) - Q_k(T)) \\
&(iii)~ \text{If } P_k(T) \geq Q_k(T), \text{ then } P_{\alpha_1, \dots, \alpha_k}(\dot{T}) \geq P_{\alpha_2, \dots, \alpha_{k+1}}(\dot{T}).
\end{align*}
\end{corollary}
Notice also, that if the substitution probability on every edge is strictly less than $\frac{r-1}{r}$ (as required by an underlying continuous-time Markov realisation of the process), then the following strict inequality result holds:  If $P_k(T) > Q_k(T)$, then $P_{\alpha_1, \dots, \alpha_k}(\dot{T}) > P_{\alpha_2, \dots, \alpha_{k+1}}(\dot{T})$.
\\ \\
In Theorem \ref{thpkqk} we consider a rooted binary tree $T$ as depicted in Fig. \ref{fig3}. 
\begin{theorem} \label{thpkqk}
Let $T$ be a rooted binary phylogenetic tree under the $N_r$ model. For $1 \leq k \leq r-1$ we have:
$P_k(T) \geq Q_k(T).$
\end{theorem}
\begin{proof} 
Since the root is assumed to be in state $\alpha_1$  and by the definition of $P_k(T)$ and $Q_k(T)$ we have that
$$ P_k(T)  =  P_{\alpha_1, \dots, \alpha_k}(T)  \mbox{ and }  Q_k(T) =  P_{\alpha_2, \dots, \alpha_{k+1}}(T).
$$
The proof is by induction on the number of leaves $n$. The inequality holds trivially for $n=1$; for  $n=2$, let $p_x, p_y$ denote the substitution probabilities on the two edges of the tree. We then have:
\begin{align*}
&P_{\alpha_1}(T) = (1-p_x)(1-p_y); &  P_{\alpha_2}(T) = \frac{p_x}{r-1}~\frac{p_y}{r-1}, \\
&P_{\alpha_1\alpha_2}(T) = (1-p_x)~\frac{p_y}{r-1}  + \frac{p_x}{r-1} ~(1-p_y); &  P_{\alpha_2\alpha_3}(T) =2~\frac{p_x}{r-1}~\frac{p_y}{r-1}.
\end{align*}
Moreover, we have:
\begin{align*}
P_{\alpha_1}(T) - P_{\alpha_2}(T) &=  (1-p_x)(1-p_y) - \frac{p_x}{r-1}~\frac{p_y}{r-1}  = 1-p_x-p_y + p_xp_y - \frac{p_x}{r-1}~\frac{p_y}{r-1} \\
&=(1- \frac{r}{r-1}~p_x)(1-\frac{r}{r-1}~p_y) +  \frac{p_x}{r-1} (1-\frac{r}{r-1}~p_y)  + \frac{p_y}{r-1} (1-\frac{r}{r-1}~p_x),
\end{align*} and
\begin{align*}
P_{\alpha_1\alpha_2}(T) - P_{\alpha_2\alpha_3}(T) &= (1-p_x)~\frac{p_y}{r-1}  + \frac{p_x}{r-1} ~(1-p_y) - 2~\frac{p_x}{r-1}~\frac{p_y}{r-1} \\
&=\frac{p_x}{r-1} (1-\frac{2p_y}{r-1}) + \frac{p_y}{r-1} (1-\frac{2p_x}{r-1}),
\end{align*} which are both non-negative, since $p_x,p_y \leq \frac{r-1}{r}$. This gives the base case of the induction. 
We now assume that the induction hypothesis holds for all trees with fewer than $n$ leaves and show that it also holds for a tree $T$ with $n$ leaves. Consider the decomposition of $T$ into its two maximal pending subtrees $T_1$ and $T_2$ and the associated  
trees $\dot{T_1}$ and $\dot{T_2}$ with a stem edge (as in Fig.~\ref{fig3}).
By the induction hypothesis,  $P_{\alpha_1, \dots, \alpha_k}(T_i) \geq  P_{\alpha_2, \dots, \alpha_{k+1}}(T_i)$ holds for $i \in \{1,2\}$. By combining  this with Corollary \ref{corrstates} (iii), we obtain: 
\begin{align}
P_{\alpha_1, \dots, \alpha_k}(\dot{T_i}) \geq  P_{\alpha_2, \dots, \alpha_{k+1}}(\dot{T_i}) \label{arg2}
\end{align}
for $i \in \{1,2\}$. Moreover, $P_k(T)$ and $Q_k(T)$ are given as follows.  Let $\omega_\alpha: =\{\alpha_1, \ldots, \alpha_k\}$ and
$\omega_\beta: =\{\alpha_2, \ldots, \alpha_{k+1}\}$, and in the following equations, $\omega_1$ and $\omega_2$ vary over
all the nonempty subsets of $\AL$ that satisfy the stated constraints under the summation signs of the following two equations: 
\begin{align}
\label{fq1}
P_k(T) = P_{\underbrace{\alpha_1, \dots, \alpha_k}_{\coloneqq \omega_\alpha}}(T) = \sum_{\omega_1 \cap \omega_2 = \omega_\alpha} P_{\omega_1}(\dot{T_1}) P_{\omega_2}(\dot{T_2}) + \sum_{\substack{\omega_1 \cap \omega_2 = \emptyset, \\ \omega_1 \cup \omega_2 =\omega_\alpha}} P_{\omega_1}(\dot{T_1}) P_{\omega_2}(\dot{T_2})
\end{align}
and
\begin{align}
\label{fq2}
Q_k(T) = P_{\underbrace{\alpha_2, \dots, \alpha_{k+1}}_{\coloneqq \omega_\beta}}(T) = \sum_{\omega_1 \cap \omega_2 = \omega_\beta} P_{\omega_1}(\dot{T_1}) P_{\omega_2}(\dot{T_2}) + \sum_{\substack{\omega_1 \cap \omega_2 = \emptyset, \\ \omega_1 \cup \omega_2 =\omega_\beta}} P_{\omega_1}(\dot{T_1}) P_{\omega_2}(\dot{T_2}).
\end{align}

To show that $P_k(T) \geq Q_k(T)$, our strategy is to show that the first term (summation)
the right-hand side of  Eqn.~(\ref{fq1}) is greater or equal to 
the first  term (summation) on the right-hand side Eqn.~(\ref{fq2}). We then 
show that same inequality also holds for the second summation term.

For any set $\omega_1^{\alpha}$ and $\omega_2^{\alpha}$ there exist 
corresponding sets $\omega_1^{\beta}$ and $\omega_2^{\beta}$. The corresponding set 
(for $i \in \{1,2\}$) is:
\begin{equation} \label{correspond}
\omega_i^{\beta} = \begin{cases}
\omega_i^{\alpha} \setminus \{\alpha_1\} \cup \{\alpha_{k+1}\} & \text{if } 
\alpha_1 \in \omega_i^{\alpha} \text{ and } \alpha_{k+1} \notin 
\omega_i^{\alpha} \\
\omega_i^{\alpha} & \text{otherwise.} 
\end{cases}
\end{equation}

For the first half of this argument, take any two sets $\omega_1^{\alpha}$ and 
$\omega_2^{\alpha}$ for which  $\omega_1^{\alpha} \cap \omega_2^{\alpha} = 
\omega_\alpha$. Note that $\alpha_1$ is contained in $\omega_1^{\alpha}$ and 
$\omega_2^{\alpha}$. Then, the corresponding sets $\omega_1^{\beta}$ and 
$\omega_2^{\beta}$ (from \eqref{correspond}) satisfy $|\omega_1^{\alpha}|=|
\omega_1^{\beta}|$ and $|\omega_2^{\alpha}|=|\omega_2^{\beta}|$ and  
$\omega_1^{\beta} \cap \omega_2^{\beta} = \omega_\beta$. Here, we consider two 
cases.

\bigskip

\noindent {\bf Case (i):}  $\alpha_1 \notin \omega_1^{\beta}$ and $\alpha_1 
\notin \omega_2^{\beta}$. \\
By  Eqn.~\eqref{arg2}, we have $P_{\omega_1^{\alpha}}(\dot{T_1}) \geq 
P_{\omega_1^{\beta}}(\dot{T_1})$ and $P_{\omega_2^{\alpha}}(\dot{T_2}) \geq 
P_{\omega_2^{\beta}}(\dot{T_2})$. Thus, $P_{\omega_1^{\alpha}}(\dot{T_1}) 
P_{\omega_2^{\alpha}}(\dot{T_2}) \geq P_{\omega_1^{\beta}}(\dot{T_1}) 
P_{\omega_2^{\beta}}(\dot{T_2})$, which completes the first case. 

\noindent{\bf Case (ii):}
 $\alpha_1$ is contained in $\omega_1^{\beta}$ or in $\omega_2^{\beta}$ (not 
both). \\  
Without loss of generality, we have $\alpha_1 \in \omega_1^{\beta}$ and 
$\alpha_1 \notin \omega_2^{\beta}$. We know that $P_{\omega_1^{\alpha}}(\dot{T_1}) 
= P_{\omega_1^{\beta}}(\dot{T_1})$ and by Eqn.~\eqref{arg2}, we have 
$P_{\omega_2^{\alpha}}(\dot{T_2}) \geq P_{\omega_2^{\beta}}(\dot{T_2})$. Thus, 
$P_{\omega_1^{\alpha}}(\dot{T_1}) P_{\omega_2^{\alpha}}(\dot{T_2}) \geq 
P_{\omega_1^{\beta}}(\dot{T_1}) P_{\omega_2^{\beta}}(\dot{T_2})$ holds.

This completes the first half of the argument.

\bigskip

We now compare the last terms on the right-hand side of the Eqns. (\ref{fq1}) and (\ref{fq2}) for 
$P_k(T)$ and $Q_k(T)$.
Take any two sets $\omega_1^{\alpha}$ and $\omega_2^{\alpha}$ for which  
$\omega_1^{\alpha} \cap \omega_2^{\alpha} = \emptyset$ and $\omega_1^{\alpha} 
\cup \omega_2^{\alpha} = \omega_\alpha$. Without loss of generality, we have 
$\alpha_1 \in \omega_1^{\alpha}$ and $\alpha_1 \notin \omega_2^{\alpha}$. 
Then, the corresponding sets $\omega_1^{\beta}$ and $\omega_2^{\beta}$ (from 
Eqn.~\eqref{correspond}) satisfy $|\omega_1^{\alpha}|=|\omega_1^{\beta}|$ and $|
\omega_2^{\alpha}|=|\omega_2^{\beta}|$ such that $\omega_1^{\beta} \cap 
\omega_2^{\beta} = \emptyset$ and $\omega_1^{\beta} \cup \omega_2^{\beta} = 
\omega_\beta$. Since $\alpha_1 \in \omega_1^{\alpha}$ and $\alpha_1 \notin 
\omega_2^{\alpha}$, we have $P_{\omega_2^{\alpha}}
(\dot{T_2})=P_{\omega_2^{\beta}}(\dot{T_2})$ and, by Eqn.~\eqref{arg2}, we have 
$P_{\omega_1^{\alpha}}(\dot{T_1}) \geq P_{\omega_1^{\beta}}(\dot{T_1})$. Thus, 
$P_{\omega_1^{\alpha}}(\dot{T_1}) P_{\omega_2^{\alpha}}(\dot{T_2}) \geq 
P_{\omega_1^{\beta}}(\dot{T_1}) P_{\omega_2^{\beta}}$ holds.

Therefore, $P_k(T)$ is greater than or equal to $Q_k(T)$ for tree $T$ by induction 
from $\dot{T_1}$ and $\dot{T_2}$.
\hfill$\Box$
\end{proof}

Combining Lemma \ref{rstatesRA} with Theorem \ref{thpkqk} gives the following corollary, which states that the reconstruction accuracy of MP under the $N_r$ model is greater or equal to $\frac{1}{r}$. In addition, note that if we assume the probabilities of change to be strictly less than $\frac{r-1}{r}$,  we can then show that $P_k(T) > Q_k(T)$ by induction on $n$ similar to the proof of Theorem \ref{thpkqk}. This gives us $RA_{\rm MP}(T) > \frac{1}{r}$. 
\begin{corollary}
For any rooted binary phylogenetic tree $T$ and the $N_r$ model, we have:
\begin{align*}
RA_{\rm MP}(T) \geq \frac{1}{r}.
\end{align*}
Moreover, this inequality is strict under a continuous-time $N_r$ model where $p_e <\frac{r-1}{r}$.
\end{corollary}

\section{A combinatorial sufficient condition for accurate ancestral state reconstruction}
\label{combo}

In this penultimate section, we present a new combinatorial property of ancestral state reconstruction using parsimony.   More precisely, we provide a sufficient condition for MP to recover the ancestral state at an interior vertex correctly from the observed states at the leaves. Note that this does not make any model assumptions (as in the previous section) as to how the character $f$ is generated -- it simply requires the state changes to be spread  sufficiently thinly in the tree as one moves way from the interior vertex.  This result complements a related (but quite different) result  from   \cite{ste05}  (Theorem 9.4.5). 

 Let $n_i$ ($i=1,2,\ldots$)  be the number of edges  descended from $v$ and separated from $v$ by $i-1$ other edges on which a substitution occurs.   Thus $n_1$ counts the number (0,1,2) of edges out of $v$ on which substitutions occur.  Note that
$n_i$ is not just a function of the tree and the character at the leaves; it depends on the actual evolution of this character on the tree. We refer to $n_i$ as the {\em substitution spectrum} of the character on the tree relative to the vertex $v$.

The following theorem can be regarded as a type of combinatorial local `safety radius' for MP to infer
the ancestral state at a given vertex correctly (even though the states at other vertices may not be correctly reconstructed).

\begin{theorem}
Consider any binary tree $T$  on any number of leaves, and any character (involving any number of states) that has evolved on this tree with a substitution spectrum relative to vertex $v$ that satisfies the inequality:
\begin{equation}
\label{boundsum}
\sum_{k \geq 1} n_k \left(\frac{1}{\sqrt{2}}\right)^k < \frac{1}{2}.
\end{equation}  
The  set of most parsimonious state at vertex $v$ estimated from the states at the leaves descending from  $v$ consists precisely of the true ancestral state at $v$  (i.e. $\FS(v) = \{\alpha\}$). 
  \end{theorem}
\begin{proof}
First observe that is sufficient to establish this result for a complete balanced binary tree $T_h$ of arbitrary height $h$, with $v$ being the root of $T_h$.
We use induction on the height $h$ of the tree. 
For $h\leq 2$ we have $n_k=0$ for all $k>2$.
Inequality~(\ref{boundsum})  ensures that $n_1=n_2=0$, in which
case all leaves are in state $\alpha$ and so the Fitch set $\FS(v)$ for $v$  is the set  $\{\alpha\}$. This establishes the result for $h \leq 2$.

For the induction step, suppose that the result holds for $T_{h-2}$ and $T_{h-1}$ and consider the tree
$T_{h}$ together with a character evolved on $T_{h}$ for which Inequality~(\ref{boundsum}) applies for vertex $v$. 
As before, this inequality ensures that none of the six edges at distance 1 or 2 descending from  $v$ have a substitution on them. 

 If $T^1$ and $T^2$ are the two maximal subtrees of $T_{h}$, then (i) each of these trees is of the type $T_{h-1}$, and
(ii) the following identity holds for all $k$:
\begin{equation}
\label{thetaeq0}
n_k = n^1_{k-1} + n^2_{k-1},
\end{equation}
 where $n^1_i$ (resp. $n^2_i$) is substitution spectrum for the character's evolution on $T^1$ and $T^2$ (note that we are  using the fact that no substitution occurs on either of the two edges outgoing from $v$, by Inequality~(\ref{boundsum})).  

Thus if we let $$p_h({\bf n}, \theta) := \sum_{k \geq 1} n_k \theta^k,$$
 where ${\bf n} =[n_k]$, then Eqn.~(\ref{thetaeq0}) allows us to write: 
$$p_h({\bf n}, \theta) = \theta \cdot [p_{h-1}({\bf n^1}, \theta)+  p_{h-1}({\bf n^2}, \theta)].$$

We can extend this argument one level further to obtain the following:
\begin{equation}
\label{thetaeq}
p_h({\bf n}, \theta) = \theta^2 \cdot [p_{h-2}({\bf n^{11}}, \theta)+  p_{h-2}({\bf n^{12}}, \theta) + p_{h-2}({\bf n^{21}}, \theta)+  p_{h-2}({\bf n^{22}}, \theta)],
\end{equation}
where ${\bf n^{ij}}$ refers the substitution spectra on the four  subtrees of type $T_{h-2}$ that are two edges descending from the vertex $v$ in $T$.
Note that in writing Eqn. (\ref{thetaeq}) we are again using the fact that Inequality~(\ref{boundsum}) precludes any substitutions in the six edges descended from $v$ and at distance at most $2$ from it.

Now put $\theta = \frac{1}{\sqrt{2}}$ in Eqn.~(\ref{thetaeq}) and let $x_{ij}:=p_{h-2}({\bf n^{ij}})$.  We then obtain:
\begin{equation}
\label{thetaeq2} p_h({\bf n}, \frac{1}{\sqrt{2}}) =\frac{1}{2}(x_{11} + x_{12} + x_{21}+ x_{22}).
\end{equation}

Since we are assuming that $p_h({\bf n}, \frac{1}{\sqrt{2}}) < \frac{1}{2}$ (by  Inequality~(\ref{boundsum})), it follows from Eqn.~(\ref{thetaeq2}) that at least three of the four terms $x_{ij}$ are strictly less than $\frac{1}{2}$, since if two of them were 
greater or equal to $\frac{1}{2}$ then $\frac{1}{2}(x_{11} + x_{12} + x_{21}+ x_{22}) \geq \frac{1}{2}.$  
By the induction hypothesis,  three (or four) of the corresponding vertices (two edges descending from $v$)  have an $\FS$ value of $\{\alpha\}$, as shown in Fig.~\ref{triangle_tree}.

\begin{figure}[htb]
\centering
\includegraphics[scale=0.8]{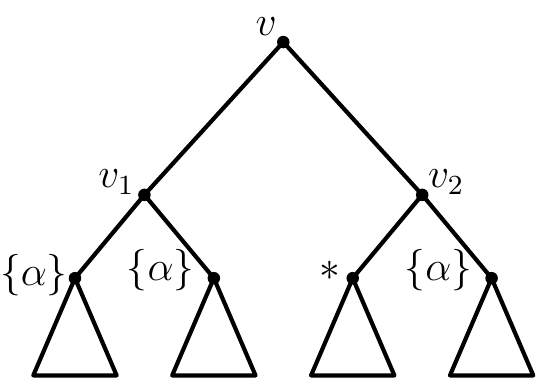}
\caption{If the Fitch sets  at the roots  of three of the four subtrees at distance 2 from $v$ in $T_h$ $(h \geq 3)$ consist of the singleton set $\{\alpha\}$, then $\FS(v)=\{\alpha\}$ as well, regardless of the Fitch set $*$ at the root of the fourth subtree. }
\label{triangle_tree}
\end{figure}

We now invoke a simple combinatorial observation: If a vertex $v$ in a  binary tree has the property that at least three vertices
that are two edges descended  from $v$ have their Fitch set $\FS$  equal to $\{\alpha\}$, then $\FS(v)= \{\alpha\}$. 
 This establishes the induction step, and thereby the theorem.
\hfill$\Box$
\end{proof}

{\bf Remark:}  An interesting question is the following: What is the smallest value of $\theta$ for which there is a constant $t$ so that
the condition $p_h({\bf n}, \theta)<t$ implies that $\FS(v)=\{\alpha\}$ for all values of $h$ and substitution spectra ${\bf n}$? We have shown that the value $\theta = \frac{1}{\sqrt{2}} \approx 0.7071$ (or any larger value) suffices, and it is known (from Theorem 2 of \cite{ste95}) that  $\theta$ cannot be smaller than the reciprocal of the golden ratio (i.e. $2/(1+\sqrt{5}) \approx 0.6180)$.

\section{Concluding comments}

Theorem~\ref{mpandcoin} demonstrated that $RA_{\rm MP}(T) \geq RA_\varphi(T)$ when $r=2$.  An interesting question is whether or not this holds more generally. This leads us to pose the following conjecture:

\begin{conjecture}
\label{conj1}
Let $T$ be a rooted binary phylogenetic tree.  Under the $N_r$ model, the reconstruction accuracy of MP is at least equal to the reconstruction accuracy of the coin-toss method:
\begin{align*}
RA_{\rm MP}(T) \geq RA_\varphi(T).
\end{align*}
\end{conjecture}
This conjecture holds for $n=2$ and all values of $r \geq 2$, as it is an exact equality in that case. 

By using Theorem~\ref{coin-toss-average}, and induction on the number of leaves,  it can be shown that Conjecture~\ref{conj1} is equivalent to the following statement:
\begin{equation}
\label{conjeq}
RA_{\rm MP}(T) \geq  \frac{1}{2} \big( RA_{\rm MP}(\dot{T_1}) + RA_{\rm MP}(\dot{T_2}) \big),
\end{equation}
where $\dot{T}_1$ and $\dot{T_2}$ are the two pending subtrees of $T$ as in Fig.~\ref{fig3}.

Inequality~(\ref{conjeq}) holds when $r=2$ since, as stated, it is equivalent to the above conjecture, and this holds when $r=2$ by
Theorem~\ref{mpandcoin}.  In the Appendix we give a direct alternative argument to justify Inequality~(\ref{conjeq}) in the case $r=2$.

\begin{acknowledgements}
The first author thanks the Ernst-Moritz-Arndt-University  Greifswald for the Landesgraduiertenf{\"o}rderung 
studentship and the German Academic Exchange Service (DAAD) for the 
DAAD-Doktorandenstipendium.  The last author thanks the New Zealand Marsden Fund (UOC-1709).  We also thank Mareike Fischer for several helpful comments.
\end{acknowledgements}
\bibliographystyle{model2-names}
\bibliography{references} 

\begin{thebibliography}{15}
\expandafter\ifx\csname natexlab\endcsname\relax\def\natexlab#1{#1}\fi
\providecommand{\url}[1]{\texttt{#1}}
\providecommand{\href}[2]{#2}
\providecommand{\path}[1]{#1}
\providecommand{\DOIprefix}{doi:}
\providecommand{\ArXivprefix}{arXiv:}
\providecommand{\URLprefix}{URL: }
\providecommand{\Pubmedprefix}{pmid:}
\providecommand{\doi}[1]{\href{http://dx.doi.org/#1}{\path{#1}}}
\providecommand{\Pubmed}[1]{\href{pmid:#1}{\path{#1}}}
\providecommand{\bibinfo}[2]{#2}
\ifx\xfnm\relax \def\xfnm[#1]{\unskip,\space#1}\fi
\bibitem[{Felsenstein(2004)}]{fels4}
\bibinfo{author}{Felsenstein, J.}, \bibinfo{year}{2004}.
\newblock \bibinfo{title}{Inferring phylogenies}.
\newblock \bibinfo{publisher}{Sinauer Press}.
\bibitem[{Fischer and Thatte(2009)}]{subset}
\bibinfo{author}{Fischer, M.}, \bibinfo{author}{Thatte, B.},
  \bibinfo{year}{2009}.
\newblock \bibinfo{title}{Maximum parsimony on subsets of taxa}.
\newblock \bibinfo{journal}{Journal of Theoretical Biology}
  \bibinfo{volume}{260}, \bibinfo{pages}{290--293}.
\newblock \DOIprefix\doi{10.1016/j.jtbi.2009.06.010}.
\bibitem[{Fitch(1971)}]{fit71}
\bibinfo{author}{Fitch, W.M.}, \bibinfo{year}{1971}.
\newblock \bibinfo{title}{Toward defining the course of evolution: Minimal
  change for a specific tree topology}.
\newblock \bibinfo{journal}{Systematic Zoology} \bibinfo{volume}{20},
  \bibinfo{pages}{406--416.}
\bibitem[{Gaschen(2002)}]{gas}
\bibinfo{author}{Gaschen, B.}, \bibinfo{year}{2002}.
\newblock \bibinfo{title}{Diversity considerations in {HIV}-1 vaccine
  selection}.
\newblock \bibinfo{journal}{Science} \bibinfo{volume}{296},
  \bibinfo{pages}{2354---2360}.
\bibitem[{Gascuel and Steel(2010)}]{gas10}
\bibinfo{author}{Gascuel, O.}, \bibinfo{author}{Steel, M.},
  \bibinfo{year}{2010}.
\newblock \bibinfo{title}{Inferring ancestral sequences in taxon-rich
  phylogenies}.
\newblock \bibinfo{journal}{Mathematical Biosciences} \bibinfo{volume}{227},
  \bibinfo{pages}{125--135}.
\bibitem[{Gascuel and Steel(2014)}]{gas14}
\bibinfo{author}{Gascuel, O.}, \bibinfo{author}{Steel, M.},
  \bibinfo{year}{2014}.
\newblock \bibinfo{title}{Predicting the ancestral character changes in a tree
  is typically easier than predicting the root state}.
\newblock \bibinfo{journal}{Systematic Biology} \bibinfo{volume}{63},
  \bibinfo{pages}{421--435}.
\bibitem[{Hartigan(1973)}]{har73}
\bibinfo{author}{Hartigan, J.A.}, \bibinfo{year}{1973}.
\newblock \bibinfo{title}{Minimum mutation fits to a given tree}.
\newblock \bibinfo{journal}{Biometrics} \bibinfo{volume}{29},
  \bibinfo{pages}{53--65}.
\bibitem[{Herbst and Fischer(2018)}]{her}
\bibinfo{author}{Herbst, L.}, \bibinfo{author}{Fischer, M.},
  \bibinfo{year}{2018}.
\newblock \bibinfo{title}{On the accuracy of ancestral sequence reconstruction
  for ultrametric trees with parsimony}.
\newblock \bibinfo{journal}{Bulletin of Mathematical Biology}
  \bibinfo{volume}{80}, \bibinfo{pages}{864--879}.
\bibitem[{Li et~al.(2008)Li, Steel and Zhang}]{moretaxa}
\bibinfo{author}{Li, G.}, \bibinfo{author}{Steel, M.}, \bibinfo{author}{Zhang,
  L.}, \bibinfo{year}{2008}.
\newblock \bibinfo{title}{More taxa are not necessarily better for the
  reconstruction of ancestral character states}.
\newblock \bibinfo{journal}{Systematic Biology} \bibinfo{volume}{57},
  \bibinfo{pages}{647--653}.
\newblock \DOIprefix\doi{10.1080/10635150802203898}.
\bibitem[{Plachetzki et~al.(2010)Plachetzki, Fong and Oakley}]{pla}
\bibinfo{author}{Plachetzki, D.C.}, \bibinfo{author}{Fong, C.R.},
  \bibinfo{author}{Oakley, T.H.}, \bibinfo{year}{2010}.
\newblock \bibinfo{title}{The evolution of phototransduction from an ancestral
  cyclic nucleotide gated pathway}.
\newblock \bibinfo{journal}{Proceedings of the Royal Society of London B:
  Biological Sciences} \bibinfo{volume}{277}, \bibinfo{pages}{1963--1969}.
\bibitem[{Steel(2016)}]{ste16}
\bibinfo{author}{Steel, M.}, \bibinfo{year}{2016}.
\newblock \bibinfo{title}{Phylogeny: Discrete and random processes in
  evolution}.
\newblock \bibinfo{publisher}{SIAM}.
\bibitem[{Steel and Penny(2005)}]{ste05}
\bibinfo{author}{Steel, M.}, \bibinfo{author}{Penny, D.}, \bibinfo{year}{2005}.
\newblock \bibinfo{title}{Maximum parsimony and the phylogenetic information in
  multi-state characters}, in: \bibinfo{editor}{Albert, V.A.} (Ed.),
  \bibinfo{booktitle}{Parsimony, phylogeny and genomics}.
  \bibinfo{publisher}{Oxford University Press}, pp. \bibinfo{pages}{163--178}.
\bibitem[{Steel and Charleston(1995)}]{ste95}
\bibinfo{author}{Steel, M.A.}, \bibinfo{author}{Charleston, M.},
  \bibinfo{year}{1995}.
\newblock \bibinfo{title}{Five surprising properties of parsimoniously colored
  trees}.
\newblock \bibinfo{journal}{Bulletin of Mathematical Biology}
  \bibinfo{volume}{57}, \bibinfo{pages}{367--375}.
\bibitem[{Tuffley and Steel(1997)}]{tuf}
\bibinfo{author}{Tuffley, C.}, \bibinfo{author}{Steel, M.},
  \bibinfo{year}{1997}.
\newblock \bibinfo{title}{Links between maximum likelihood and maximum
  parsimony under a simple model of site substitution}.
\newblock \bibinfo{journal}{Bulletin of Mathematical Biology}
  \bibinfo{volume}{59}, \bibinfo{pages}{581--607}.
\bibitem[{Zhang et~al.(2010)Zhang, Shen, Yang and Li}]{ra1}
\bibinfo{author}{Zhang, L.}, \bibinfo{author}{Shen, J.}, \bibinfo{author}{Yang,
  J.}, \bibinfo{author}{Li, G.}, \bibinfo{year}{2010}.
\newblock \bibinfo{title}{Analyzing the {F}itch method for reconstructing
  ancestral states on ultrametric phylogenetic trees}.
\newblock \bibinfo{journal}{Bulletin of Mathematical Biology}
  \bibinfo{volume}{72}, \bibinfo{pages}{1760--1782}.
\newblock \DOIprefix\doi{10.1007/s11538-010-9505-8}.

\end{thebibliography}

\section{Appendix: Direct proof of  Inequality~(\ref{conjeq}) when $r=2$}
\begin{proof}
For the $N_2$ model, $RA_{\rm MP}(T) = P_{\alpha}(T) + \frac{1}{2} P_{\alpha\beta}(T)$. Thus,
\begin{align}
RA_{\rm MP}(T)  &= P_{\alpha}(\dot{T_1})P_{\alpha}(\dot{T_2}) + P_{\alpha}(\dot{T_1})P_{\alpha\beta}(\dot{T_2}) + P_{\alpha\beta}(\dot{T_1})P_{\alpha}(\dot{T_2}) \nonumber \\
&\qquad + \frac{1}{2} \big(P_{\alpha\beta}(\dot{T_1})P_{\alpha\beta}(\dot{T_2}) + P_{\alpha}(\dot{T_1})P_{\beta}(\dot{T_2}) + P_{\beta}(\dot{T_1})P_{\alpha}(\dot{T_2})\big). \label{(ra)}
\end{align}
Moreover:
\begin{align}
&\frac{1}{2} \big( RA_{\rm MP}(\dot{T_1}) + RA_{\rm MP}(\dot{T_2}) \big) = \frac{1}{2} \big( P_{\alpha}(\dot{T_1}) + \frac{1}{2} P_{\alpha\beta}(\dot{T_1}) + P_{\alpha}(\dot{T_2}) + \frac{1}{2} P_{\alpha\beta}(\dot{T_2}) \big) \nonumber \\
&= \frac{1}{2} \big( P_{\alpha}(\dot{T_1}) + P_{\alpha}(\dot{T_2}) \big) + \frac{1}{4} \big( P_{\alpha\beta}(\dot{T_1}) + P_{\alpha\beta}(\dot{T_2}) \big) \nonumber \\
&= \frac{1}{2} \big( P_{\alpha}(\dot{T_1}) (P_{\alpha}(\dot{T_2})+P_{\beta}(\dot{T_2})+P_{\alpha\beta}(\dot{T_2})) + P_{\alpha}(\dot{T_2}) (P_{\alpha}(\dot{T_1})+P_{\beta}(\dot{T_1})+P_{\alpha\beta}(\dot{T_1})) \big) \nonumber \\
&\qquad+ \frac{1}{4} \big( P_{\alpha\beta}(\dot{T_1}) + P_{\alpha\beta}(\dot{T_2}) \big)
\text{(by the law of total probability)} \nonumber \\
&=\frac{1}{2} \big( 2P_{\alpha}(\dot{T_1}) P_{\alpha}(\dot{T_2}) +P_{\alpha}(\dot{T_1}) P_{\beta}(\dot{T_2}) + P_{\alpha}(\dot{T_1}) P_{\alpha\beta}(\dot{T_2}) + P_{\beta}(\dot{T_1})P_{\alpha}(\dot{T_2}) + P_{\alpha\beta}(\dot{T_1})P_{\alpha}(\dot{T_2}) \big) \nonumber \\
&\qquad+ \frac{1}{4} \big( P_{\alpha\beta}(\dot{T_1}) + P_{\alpha\beta}(\dot{T_2}) \big) \label{(ra)/2}.
\end{align}
In order to show that $RA_{\rm MP}(T) \geq \frac{1}{2} \big( RA_{\rm MP}(\dot{T_1}) + RA_{\rm MP}(\dot{T_2}) \big)$, we establish the following inequality:   $$RA_{\rm MP}(T) - \frac{1}{2} \big( RA_{\rm MP}(\dot{T_1}) + RA_{\rm MP}(\dot{T_2}) \big) \geq 0.$$
 By \eqref{(ra)} and \eqref{(ra)/2} we have:
\begin{align*}
&RA_{\rm MP}(T) - \frac{1}{2} \big( RA_{\rm MP}(\dot{T_1}) + RA_{\rm MP}(\dot{T_2}) \big) \\
&= P_{\alpha}(\dot{T_1})P_{\alpha}(\dot{T_2}) + P_{\alpha}(\dot{T_1})P_{\alpha\beta}(\dot{T_2}) + P_{\alpha\beta}(\dot{T_1})P_{\alpha}(\dot{T_2}) \nonumber \\
&\qquad + \frac{1}{2} \big(P_{\alpha\beta}(\dot{T_1})P_{\alpha\beta}(\dot{T_2}) + P_{\alpha}(\dot{T_1})P_{\beta}(\dot{T_2}) + P_{\beta}(\dot{T_1})P_{\alpha}(\dot{T_2})\big) \\
&\qquad - \frac{1}{2} \big( 2P_{\alpha}(\dot{T_1}) P_{\alpha}(\dot{T_2}) +P_{\alpha}(\dot{T_1}) P_{\beta}(\dot{T_2}) + P_{\alpha}(\dot{T_1}) P_{\alpha\beta}(\dot{T_2}) + P_{\beta}(\dot{T_1})P_{\alpha}(\dot{T_2}) + P_{\alpha\beta}(\dot{T_1})P_{\alpha}(\dot{T_2}) \big) \\
&\qquad - \frac{1}{4} \big( P_{\alpha\beta}(\dot{T_1}) + P_{\alpha\beta}(\dot{T_2}) \big) \\
&= \frac{1}{2}P_{\alpha}(\dot{T_1})P_{\alpha\beta}(\dot{T_2}) + \frac{1}{2}P_{\alpha\beta}(\dot{T_1})P_{\alpha}(\dot{T_2}) + \frac{1}{2} P_{\alpha\beta}(\dot{T_1})P_{\alpha\beta}(\dot{T_2}) - \frac{1}{4} \big( P_{\alpha\beta}(\dot{T_1}) + P_{\alpha\beta}(\dot{T_2}) \big) \\
&= \frac{1}{2} P_{\alpha\beta}(\dot{T_1}) \big(P_{\alpha}(\dot{T_2}) + \frac{1}{2}P_{\alpha\beta}(\dot{T_2}) - \frac{1}{2} \big) + \frac{1}{2} P_{\alpha\beta}(\dot{T_2}) \big(P_{\alpha}(\dot{T_1}) + \frac{1}{2}P_{\alpha\beta}(\dot{T_1}) - \frac{1}{2} \big) \\
&= \frac{1}{2} P_{\alpha\beta}(\dot{T_1}) \big( RA_{\rm MP}(\dot{T_2}) - \frac{1}{2} \big) + \frac{1}{2} P_{\alpha\beta}(\dot{T_2}) \big(RA_{\rm MP}(\dot{T_1}) - \frac{1}{2} \big). \\
\end{align*}
This last expression is non-negative because the reconstruction accuracy under the $N_2$ model is greater or equal to $\frac{1}{2}$ by Corollary~\ref{1-p}, and (by  Lemma~\ref{lemhelps}),
$RA_{\rm MP}(\dot{T}) \geq \frac{1}{2}$ if and only if $RA_{\rm MP}(T) \geq \frac{1}{2}$.
\hfill$\Box$
\end{proof}
\end{document}